\documentclass[submission,copyright,creativecommons]{eptcs}
\providecommand{\event}{GandALF 2019} 
\usepackage{underscore}           

\usepackage[english]{babel}
\usepackage[T1]{fontenc} 
\usepackage[utf8]{inputenc}

\usepackage{amsmath,amssymb,amsthm}

\usepackage{tikz}
\usetikzlibrary{arrows, automata, shapes, positioning, decorations}

\usepackage{rotating}
\usepackage{graphicx}

\title{State Complexity of the Multiples of the Thue-Morse Set}

\author{Émilie Charlier
\institute{University of Liège\\
Belgium} 
\email{echarlier@uliege.be}
\and Célia Cisternino
\institute{University of Liège\\
Belgium} 
\email{ccisternino@uliege.be} 
\and Adeline Massuir
\institute{University of Liège\\
Belgium} 
\email{a.massuir@uliege.be}
}

\def\titlerunning{State Complexity of the Multiples of the Thue-Morse Set}
\def\authorrunning{É. Charlier, C. Cisternino \& A. Massuir}

\theoremstyle{plain}
\newtheorem{theorem}{Theorem}
\newtheorem{lemma}[theorem]{Lemma}
\newtheorem{corollary}[theorem]{Corollary}
\newtheorem{proposition}[theorem]{Proposition}

\theoremstyle{definition}
\newtheorem{definition}[theorem]{Definition}
\newtheorem{example}[theorem]{Example}
\newtheorem{remark}[theorem]{Remark}

\newcommand{\N}{\mathbb{N}}
\newcommand{\rep}{\mathrm{rep}}
\newcommand{\val}{\mathrm{val}}

\newcommand{\A}{\mathcal{A}}
\newcommand{\T}{\mathcal{T}}
\newcommand{\andrm}{\ {\rm and}\ }

\usepackage{pgf, tikz}
\usetikzlibrary{arrows, automata, shapes, positioning, decorations}

\definecolor{bientotlafin}{RGB}{142, 162, 198}
\definecolor{green}{RGB}{31,160,85}
\definecolor{vert}{RGB}{0,255,127}

\begin{document}
\maketitle

\begin{abstract}
The Thue-Morse set $\mathcal{T}$ is the set of those non-negative integers whose binary expansions have an even number of $1$.  The name of this set comes from the fact that its characteristic sequence is given by the famous Thue-Morse word ${\tt abbabaabbaababba\cdots}$, which is the fixed point starting with ${\tt a}$ 
of the word morphism ${\tt a\mapsto ab,b\mapsto ba}$. The numbers in $\mathcal{T}$ are sometimes called the {\em evil numbers}. We obtain an exact formula for the state complexity (i.e.\ the number of states of its minimal automaton) of the multiplication by a constant of the Thue-Morse set  with respect to any integer base $b$ which is a power of $2$. Our proof is constructive and we are able to explicitly provide the minimal automaton of the language of all $2^p$-expansions of the set $m\mathcal{T}$ for any positive integers $m$ and $p$. The used method is general for any $b$-recognizable set of integers. 
As an application, we obtain a decision procedure running in quadratic time for the problem of deciding whether a given $2^p$-recognizable set is equal to some multiple of the Thue-Morse set. 
\end{abstract}

\section{Introduction}

A subset $X$ of $\N$ is said to be {\em $b$-recognizable} if the base-$b$ expansions of the elements of $X$ form a regular language. The famous theorem of Cobham tells us that any non-trivial property of numbers are dependent on the base we choose: the only sets that are $b$-recognizable for all bases $b$ are the finite unions of arithmetic progressions \cite{Cobham:1969}. Inspired by this seminal result, many descriptions of $b$-recognizable sets were given, e.g.\ morphic, algebraic and logical characterizations \cite{Boigelot&Rassart&Wolper:1998,Bruyere&Hansel&Michaux&Villemaire:1994,Cobham:1972}, extensions of these to systems based on a Pisot number \cite{Bruyere&Hansel:1997}, the normalization map \cite{Frougny:1992} or the possible growth functions \cite{Charlier&Rampersad:2011,Eilenberg:1974}. For more on $b$-recognizable sets, we refer to the surveys \cite{Allouche&Shallit:2003,Bruyere&Hansel&Michaux&Villemaire:1994,Charlier:2018,Eilenberg:1974,Frougny&Sakarovitch:2010,Rigo:2014}.

In particular, as mentioned above, these sets have been characterized in terms of logic. More precisely, a subset of $\N$ (and more generally of $\N^d$) is $b$-recognizable if and only if it is definable by a first-order formula of the structure $\langle\N,+,V_b\rangle$ where $V_b$ is the base-dependent functional predicate that associates with a natural $n$ the highest power of $b$ dividing $n$. Since the finite unions of arithmetic progressions are precisely the subsets of $\N$ that are definable by first order formulas in the Presburger arithmetic $\langle\mathbb{N},+\rangle$, this characterization provides us with a logical interpretation of Cobham’s theorem. In addition, this result turned out to be a powerful tool for showing that many properties of $b$-automatic sequences are decidable and, further, that many enumeration problems of $b$-automatic sequences can be described by $b$-regular sequences in the sense of Allouche and Shallit \cite{Allouche&Shallit:1992,Allouche&Shallit:2003,Charlier&Rampersad&Shallit:2012}. 

In the context of Cobham's theorem, the following question is natural 
and has received a constant attention during the last 30 years: given an automaton accepting the language of the base-$b$ expansions of a set $X\subseteq\N$, is it decidable whether $X$ is  a finite union of arithmetic progressions? Several authors gave decision procedures for this problem \cite{Allouche&Rampersad&Shallit:2009,Bruyere&Hansel&Michaux&Villemaire:1994,Honkala:1986,Leroux:2005,Marsault&Sakarovitch:2013}. Moreover, a multidimensional version of this problem was shown to be decidable in a beautiful way based on logical methods \cite{Bruyere&Hansel&Michaux&Villemaire:1994,Muchnik:2003}. 

With any set of integers $X$ is naturally associated an infinite word, which is its characteristic sequence $\chi_X\colon n\mapsto 1$ if $n\in X,\  n\mapsto 0$ otherwise. Thus, to a finite union of arithmetic progressions corresponds an ultimately periodic infinite word. Therefore, the HD0L ultimate periodicity problem consisting in deciding whether a given morphic word (i.e.\ the image under a coding of the fixed point of a morphism) is ultimately periodic is a generalization of the periodicity problem for $b$-recognizable sets mentioned in the previous paragraph. The HD0L ultimate periodicity problem
was shown to be decidable in its full generality \cite{Durand:2013,Mitrofanov:2013}. The proofs rely on return words, primitive substitutions or evolution of Rauzy graphs. However, these methods do not provide algorithms that could be easily implemented and the corresponding time complexity is very high. In addition, they do not allow us to obtain an algorithm for the multidimensional generalization of the periodicity problem, i.e.\ the problem of deciding whether a $b$-recognizable subset of $\N^d$ is definable within the Presburger arithmetic $\langle\N,+\rangle$. Therefore, a better understanding of the inner structure of automata arising from number systems remains a powerful tool to obtain efficient decision procedures. 

The general idea is as follows. Suppose that $\mathcal{L}=\{L_i\colon i\in \N\}$ is a collection of languages and that we want to decide whether some particular language $L$ belongs to $\mathcal{L}$. Now, suppose that we are able to explicitly give a lower bound on the state complexities of the languages in $\mathcal{L}$, i.e.\ for each given $N$, we can effectively produce a bound $B(N)$ such that for all $i> B(N)$, the state complexity of $L_i$ is greater than $N$. Then the announced problem is decidable: if $k$ is the state complexity of the given language $L$, then only the finitely many languages $L_0,\ldots,L_{B(k)}$ have to be compared with $L$.

The state complexity of a $b$-recognizable set (i.e.\ the number of states of the minimal automaton accepting the $b$-expansions of its elements) is closely related to the length
of the logical formula describing this set. Short formulas are crucial in order to produce efficient mechanical proofs by using for example the Walnut software \cite{Mousavi:2015,Shallit:2015}. There are several ways to improve the previous decision procedure. One of them if to use precise knowledge of the stucture of the involved automata. This idea was successfully used in the papers \cite{BMMR,Marsault&Sakarovitch:2013}. 
In \cite{Charlier&Rampersad&Rigo&Waxweiler:2011},  the structure of automata accepting the greedy expansions of $m\N$ for a wide class of non-standard numeration systems, and in particular, estimations of 
 the state complexity of $m\N$ are given. 
Another way of improving this procedure is to have at our disposal the exact state complexities of the languages in $\mathcal{L}$. Finding an exact formula is a much more difficult problem than finding good estimates. However, some results in this direction are known. For instance, it is proved in \cite{Charlier&Rampersad&Rigo&Waxweiler:2011} that for the Zeckendorf numeration system (i.e.\ based on the Fibonacci numbers), the state complexity of $m\N$ is exactly $2m^2$. 
A complete description of the minimal automaton recognizing $m\N$ in any integer base $b$ was given in~\cite{Alexeev:2004} and the state complexity of $m\N$ with respect to the base $b$ is shown to be exactly 
\begin{equation}\label{eq:alexeev}
 \frac{m}{\gcd(m,b^N)}+\sum_{t=0}^{N-1} \frac{b^t}{\gcd(m,b^t)}
\end{equation}
where $N$ is the smallest integer $\alpha$ such that $\frac{m-b^\alpha}{\gcd(m,b^\alpha)}<\frac{m}{\gcd(m,b^{\alpha+1})}$. 

For all the above mentioned reasons, the study of the state complexity of 
$b$-recognizable sets deserves special interest.
In the present work, we propose ourselves to initiate a study of the state complexity of the multiplication by a constant of recognizable subsets $X$ of $\N$. In doing so, we aim at generalizing the previous framework concerning the case $X=\N$ only.  Our study starts with the Thue-Morse set $\T$ of the so-called {\em evil numbers} \cite{Allouche:2015}, i.e.\ the natural numbers whose base-$2$ expansion contains an even number of occurrences of the digit $1$. The characteristic sequence of this set corresponds to the ubiquitous Thue-Morse word ${\tt abbabaabbaababba\cdots}$, which is the fixed point starting with ${\tt a}$ of the morphism ${\tt a\mapsto ab,b\mapsto ba}$. This infinite word is one of the archetypical aperiodic automatic words. 
Therefore, the set $\T$ seems to be a natural candidate to start with. The goal of this work is to provide a complete characterization of the minimal automata recognizing the sets $m\T$ for any multiple $m$ and any base $b$ which is a power of $2$ (other bases are not relevant with the choice of the Thue-Morse set in view of Cobham's theorem).

This paper has the following organization. In Section~\ref{sec:basics}, we recall the background that is necessary to tackle our problem. In Section~\ref{sec:method}, we state our main result and expose the method that will be carried out for its proof. More precisely, we present the steps of our construction of the minimal automaton accepting the base-$2^p$ expansions of the elements of $m\T$ for any positive integers $m$ and $p$. Sections~\ref{sec:aut-T} to~\ref{sec:projected-product} are devoted to build each needed intermediate automata. Thus, at the end of Section~\ref{sec:projected-product}, we are provided with an automaton recognizing the desired language. At each step of the construction, we study the properties of the built automata that will be needed for proving the announced state complexity result. The minimization procedure of the last automaton is handled in Section~\ref{sec:minimization}. This part is the most technical one and it deeply relies on the properties of the intermediate automata proved in the previous sections. Finally, in Section~\ref{sec:perspectives}, we discuss future work and give three related open problems. 
Due to lack of space, this paper does not contain full proofs of our results. Nevertheless, all the missing details can be found in the arXiv platform \cite{arxiv-version}.

\section{Basics}
\label{sec:basics}

In this text, we use the usual definitions and notation (alphabet, letter, word, language, free monoid, automaton, etc.) of formal language theory  \cite{Lothaire1997,Sakarovitch2009}. 
Nevertheless, let us give a few definitions and properties that will be central in this work. The length of a finite word $w$ is denoted by $|w|$ and the number of occurrences of a letter $a$ in $w$ is denoted by $|w|_a$. The empty word is denoted by $\varepsilon$. A {\em regular language} is a language which is accepted by a finite automaton. For $L\subseteq A^*$ and $w\in A^*$, the {\em (left) quotient} of $L$ by $w$ is the language $w^{{-}1}L=\{u\in A^*\colon wu\in L\}$. As is well known, a language $L$ over an alphabet $A$ is regular if and only if it has finitely many quotients, that is, the set of languages $\{w^{{-}1}L\colon w\in A^*\}$ is finite. The {\em state complexity} of a regular language is the number of its quotients. It corresponds to the number of states of its minimal automaton. The following characterization of minimal automata will be used several times in this work: a deterministic finite automaton (or DFA for short) is minimal if and only if it is reduced and accessible. Recall that a DFA is reduced if the languages accepted from distinct states are distinct and that a DFA is accessible if every state can be reached from the initial state. The language accepted from a state $q$ is denoted by $L_q$. Thus, the language accepted by a DFA is the language accepted from its initial state (we always consider automata having a single initial state).

In what follows we will need a notion that is somewhat stronger than that of reduced DFAs. We say that a DFA has {\em disjoint states} if the languages accepted from distinct states are disjoint: for distinct states $p$ and $q$, we have $L_p\cap L_q=\emptyset$. A state $q$ is said to be {\em co-accessible} if $L_q\ne \emptyset$ and, by extension, an automaton is said to be {\em co-accessible} if all its states are co-accessible. Thus, any co-accessible DFA having disjoint states is reduced.

Now, let us give some background on numeration systems. Let $b\in\N_{\ge2}$. The {\em $b$-expansion} of a positive integer $n$, which is denoted by $\rep_b (n)$, is the finite word $c_{\ell{-}1}\cdots c_0$ over the alphabet $A_b=[\![0,b{-}1]\!]$ defined by $n=\sum_{j=0}^{\ell{-}1} c_j b^j,  c_{\ell{-}1}\ne 0$. Note that here and throughout the text, an interval of integers $\{m,m+1,\ldots,n\}$ is denoted by $[\![m,n]\!]$. The {\em $b$-expansion} of $0$ is the empty word: $\rep_b(0)=\varepsilon$. The number $b$ is called the {\em base} of the numeration. Conversely, for a word $w=c_{\ell{-}1}\cdots c_0$ over the alphabet $A_b$, we write $\val_b(w)=\sum_{j=0}^{\ell{-}1} c_j b^j$. Thus we have $\rep_b\colon\N\to A_b^*$ and $\val_b\colon A_b^*\to \N$. For all subsets $X$ of $\N$, we have $\val^{-1}_b(X)=0^*\rep_b(X)$. A subset $X$ of $\N$ is said to be {\em $b$-recognizable} if the language $\rep_b (X)$ is regular. It is of course equivalent to ask that the language $\val_b^{-1}(X)$ is regular. In what follows, we will always consider automata accepting $\val_b^{-1}(X)$ instead of $\rep_b(X)$. The {\em state complexity} of a $b$-recognizable subset $X$ of $\N$ {\em with respect to the base $b$} is the state complexity of the language $\val_b^{-1}(X)$. 

We will need to represent not only natural numbers, but also pairs of natural numbers. If $u=u_1\cdots u_n\in A^*$ and $v=v_1\cdots v_n\in B^*$ are words of the same length $n$, then we use the notation $(u,v)$ to designate the word $(u_1,v_1)\cdots (u_n,v_n)$ of length $n$ over the alphabet $A\times B$. 
For $(m,n)\in\N^2$, we write
\[
	\rep_b(m,n)=(0^{\ell-|\rep_b(m)|}\rep_b(m),0^{\ell-|\rep_b(n)|}\rep_b(n))
\] 
where $\ell=\max\{|\rep_b(m)|,|\rep_b(n)|\}$. Finally, for a subset $X$ of $\N^2$, we write $\val_{b}^{-1}(X)=(0,0)^*\rep_b(X)$.

\section{Main result and method}
\label{sec:method}

The Thue-Morse set, which we denote by $\T$, is the set of all natural numbers whose base-$2$ expansions contain an even number of occurrences of $1$:
\[
	\T= \{ n \in \N \colon |\rep_2 (n)|_1 \in 2\N\}.
\]
Note that the numbers in $\T$ are sometimes called {\em evil} and the numbers in $\N \setminus \T$ are said to be {\em odious} \cite{Allouche:2015}. The set $\T$ is clearly $2$-recognizable. More precisely, it is $2^p$-recognizable for all $p\in\N_{\ge1}$ and is not $b$-recognizable for any other base $b$. For example, an automaton recognizing $\T$ in base $4$ is depicted in the left part of Figure~\ref{fig:aut-TM-2-4-couples}. This is a consequence of the theorem of Cobham. Two positive integers are said to be {\em multiplicatively independent} if their only common integer power is $1$.

\begin{theorem}[Cobham \cite{Cobham:1969}] \ 
\begin{itemize}
\item Let $b,b'$ be two multiplicatively independent bases. Then a subset of $\N$ is both $b$-recognizable and $b'$-recognizable if and only if it is a finite union of arithmetic progressions.
\item Let $b,b'$ be two multiplicatively dependent bases. Then a subset of $\N$ is $b$-recognizable if and only if it is $b'$-recognizable.
\end{itemize}
\end{theorem}

We introduce the following notation: for $X\in\{T,B\}$ and $n \in \N$, we define
\[
	X_n = \begin{cases}
			X & \text{if } n \in \T \\
			\overline{X} & \text{else}  
			\end{cases}
\]
where $\overline{T}=B$ and $\overline{B}=T$. It is easily seen that for each $p\in\N_{\ge1}$, the language $\val^{-1}_{2^p}(\T)$ is accepted by the DFA $(\{T,B\},T,T,A_{2^p},\delta)$ where for all $X\in\{T,B\}$ and all $a\in A_{2^p}$, $\delta(X,a)=X_a$. 

The following proposition is well known; for example see~\cite{Bruyere&Hansel&Michaux&Villemaire:1994}. 

\begin{proposition}
Let $b\in\N_{\ge 2}$ and $m\in\N$. If $X$ is $b$-recognizable, then so is $mX$. Otherwise stated, multiplication by a constant preserves $b$-recognizability.
\end{proposition}

In particular, for any positive integers $m$ and $p$, the set $m\T$ is $2^p$-recognizable. The aim of this work is to prove the following result.

\begin{theorem} \label{thm:main}
Let $m$ and $p$ be positive integers. Then the state complexity of $m\T$ with respect to the base $2^p$ is equal to 
\[
	2k+\left\lceil \frac zp\right\rceil
\]
where $z$ is the highest power of $2$ dividing $m$ and $k$ is the odd part of $m$, i.e.\ $z$ and $k$ are the unique integers such that $m=k2^z$ with $k$ odd.
\end{theorem}

Our proof of Theorem~\ref{thm:main} is constructive. In order to describe the minimal DFA of $\val_{2^p}^{-1}(m\T)$, we will successively construct  several automata. First, we build a DFA $\A_{\T,2^p}$ accepting the language $\val_{2^p}^{-1}(\T\times \N)$. Then we build a DFA $\A_{m,b}$ accepting the language $\val_b^{-1}\big(\{(n,mn)\colon n\in \N\}\big)$.
Note that we do the latter step for any integer base $b$ and not only for powers of $2$.
Next, we consider the product automaton $\A_{m,2^p}\times\A_{\T,2^p}$. This DFA accepts the language $\val_{2^p}^{-1}\big(\{(t,mt)\colon t \in \T\}\big)$. Finally, a finite automaton $\Pi(\A_{m,2^p}\times\A_{\T,2^p})$ accepting $\val_{2^p}^{-1}(m\T)$ is obtained by projecting the label of each transition in $\A_{m,2^p}\times\A_{\T,2^p}$ onto its second component.
At each step of our construction, we check that the automaton under consideration is minimal (and hence deterministic) and the ultimate step precisely consists in a minimization procedure.

From now on, we fix some positive integers $m$ and $p$. We also let $z$ and $k$ be the unique integers such that $m=k2^z$ with $k$ odd.

\section{The automaton $\A_{\T,2^p}$}
\label{sec:aut-T}

In this section, we construct a DFA  $\A_{\T,2^p}$ accepting $\val^{-1}_{2^p}(\T\times\N)$.
This DFA is a modified version of the automaton accepting $\val^{-1}_{2^p}(\T)$ defined in the previous section. Namely, we replace each transition labeled by $a\in A_{2^p}$ by $2^p$ copies of itself labeled by $(a,b)$, for each $b\in A_{2^p}$. Formally, 
\[
	\A_{\T,2^p}=(\{T,B\},T,T,A_{2^p}\times A_{2^p},\delta_{\T,2^p})
\]
where, for all $X\in\{T,B\}$ and all $a,b\in A_{2^p}$, we have $\delta_{\T,2^p}(X,(a,b))=X_a$. (The letters $B$ and $T$ were not chosen arbitrarily: $B$ is for “bottom“ whereas the letter $T$ refers to both “top” and “Thue-Morse”.) The automaton $\A_{\T,4}$ (i.e.\ for $p=2$) is depicted in the right part of Figure~\ref{fig:aut-TM-2-4-couples}.

\begin{figure}[htb]
\vspace*{-1.3cm}
\begin{minipage}[c]{0.35\linewidth}
\centering
\begin{tikzpicture}
\tikzstyle{every node}=[shape=circle, fill=none, draw=black,
minimum size=20pt, inner sep=2pt,scale=0.8]
\node(1) at (0,0) {$T$};
\node(2) at (0,-1.5) {$B$};
\tikzstyle{every node}=[shape=circle, fill=none, draw=black,
minimum size=15pt, inner sep=2pt,scale=0.8]
\node(2f) at (0,0) {};
\tikzstyle{every path}=[color=black, line width=0.5 pt]
\tikzstyle{every node}=[shape=circle, minimum size=5pt, inner sep=2pt]
\draw [->] (-1,0) to node {} (1); 
\draw [->] (1) to [loop above] node [above=-0.15cm,scale=0.8] {$0,3$} (1);
\draw [->] (2) to [loop below] node [below=-0.15cm,scale=0.8] {$0,3$} (2);
\draw [->] (1) to [bend right=30] node [left,scale=0.8] {$1,2$} (2);
\draw [->] (2) to [bend right=30] node [right,scale=0.8] {$1,2$} (1);
\end{tikzpicture}
\end{minipage}
\begin{minipage}[c]{0.6\linewidth}
\centering
\begin{tikzpicture}
\tikzstyle{every node}=[shape=circle, fill=none, draw=black,
minimum size=20pt, inner sep=2pt,scale=0.8]
\node(1) at (0,0) {$T$};
\node(2) at (0,-1.75) {$B$};
\tikzstyle{every node}=[shape=circle, fill=none, draw=black,
minimum size=15pt, inner sep=2pt,scale=0.8]
\node(2f) at (0,0) {};
\tikzstyle{every path}=[color=black, line width=0.5 pt]
\tikzstyle{every node}=[shape=circle, minimum size=5pt, inner sep=2pt]
\draw [->] (-1,0) to node {} (1); 
\draw [->] (1) to [loop above] node [above=-1.3cm,scale=0.8] {\begin{tabular}{c}
$(0,0),(0,1),(0,2),(0,3)$ \\
$(3,0),(3,1),(3,2),(3,3)$ \\
\end{tabular}} (1);
\draw [->] (2) to [loop below] node [below=-1.3cm,scale=0.8] {\begin{tabular}{c}
$(0,0),(0,1),(0,2),(0,3)$ \\
$(3,0),(3,1),(3,2),(3,3)$ \\
\end{tabular}} (2);
\draw [->] (1) to [bend right=30] node [left=-0.1cm,scale=0.8] {\begin{tabular}{c}
$(1,0),(1,1),(1,2),(1,3)$ \\
$(2,0),(2,1),(2,2),(2,3)$ \\
\end{tabular}} (2);
\draw [->] (2) to [bend right=30] node [right=-0.1cm,scale=0.8] {\begin{tabular}{c}
$(1,0),(1,1),(1,2),(1,3)$ \\
$(2,0),(2,1),(2,2),(2,3)$ \\
\end{tabular}} (1);
\end{tikzpicture}
\end{minipage}
\vspace*{-1.5cm}
\caption{The minimal automaton recognizing the Thue-Morse set in base $4$ (left) and the automaton $\A_{\T,4}$ (right).}
\label{fig:aut-TM-2-4-couples}
\end{figure}
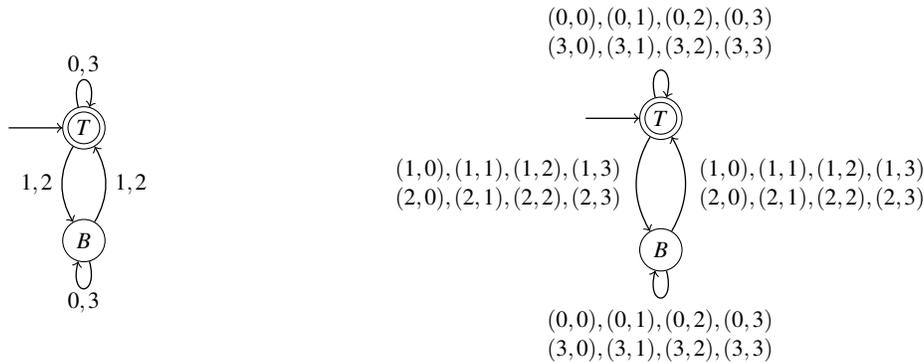

Proofs of the following two lemmas are easy verifications. 

\begin{lemma}
The automaton $\A_{\T,2^p}$ is complete, accessible, co-accessible and has disjoint states. In particular, it is the minimal automaton of $\val_{2^p}^{-1}(\T\times \N)$.
\end{lemma}

\begin{lemma}\label{lem:transitionsTM}
For all $X \in \{T,B\}$ and $(u,v)\in (A_{2^p}\times A_{2^p})^*$, we have 
\[
	\delta_{\T,2^p}(X,(u,v))=X_{\val_{2^p}(u)}.
\]
\end{lemma}

\section{The automaton $\A_{m,b}$}
\label{sec:aut-m}

In this section, we consider an arbitrary integer base $b$. Let 
\[
	\A_{m,b}=([\![0,m{-}1]\!],0,0,A_b\times A_b,\delta_{m,b})
\]
where the (partial) transition function $\delta_{m,b}$ is defined as follows: for each $i,j\in[\![0,m{-}1]\!]$ and each $d,e\in A_b$, we set
\[
	\delta_{m,b}(i,(d,e))=j  \iff  bi + e = md + j.
\]
This DFA accepts the language $\val^{-1}_b(\{(n,mn)\colon n \in \N \}$. We refer the interested reader to \cite{Waxweiler2009}. 
For example, the automaton $\A_{6,4}$ is depicted in Figure~\ref{fig:A-6,4}. 

\begin{figure}[htb]
\centering
\begin{tikzpicture}[scale=0.7]
\tikzstyle{every node}=[shape=circle, fill=none, draw=black,
minimum size=30pt, inner sep=2pt,scale=0.7]
\node(0) at (0,6.5) {$0$};
\node(1) at (3.5,6.5) {$1$};
\node(2) at (7,6.5) {$2$};
\node(3) at (10.5,6.5) {$3$};
\node(4) at (14,6.5) {$4$};
\node(5) at (17.5,6.5) {$5$};

\tikzstyle{every node}=[shape=circle, fill=none, draw=black,
minimum size=25pt, inner sep=2pt,scale=0.7]
\node at (0,6.5) {};

\tikzstyle{etiquettedebut}=[very near start,rectangle,fill=black!20,scale=0.8]
\tikzstyle{etiquettemilieu}=[midway,rectangle,fill=black!20,scale=0.8]
\tikzstyle{every path}=[color=black, line width=0.5 pt,scale=0.8]
\tikzstyle{every node}=[shape=circle, minimum size=5pt, inner sep=2pt,scale=0.8]

\draw [->] (-1.5,8.1) to node {} (0); 
\draw [->] (0) to [loop above] node [left,rectangle,fill=black!20,scale=0.8] {$(0,0)$} (0);
\draw [->] (1) to [loop below] node [left,rectangle,fill=black!20,scale=0.8] {$(1,3)$} (1);
\draw [->] (2) to [loop left] node [left,pos=0.2,rectangle,fill=black!20,scale=0.8] {$(1,0)$} (2);
\draw [->] (3) to [loop below] node [right,pos=0.3,rectangle,fill=black!20,scale=0.8] {$(2,3)$} (3);
\draw [->] (4) to [loop above] node [left,rectangle,fill=black!20,scale=0.8] {$(2,0)$} (4);
\draw [->] (5) to [loop above] node [left,rectangle,fill=black!20,scale=0.8] {$(3,3)$} (5);
\draw [->] (0) to [bend left=17] node [sloped,etiquettemilieu] {$(0,1)$} (1);
\draw [->] (0) to [bend left=30] node [sloped,etiquettemilieu] {$(0,2)$} (2);
\draw [->] (0) to [bend left=45] node [sloped,etiquettemilieu] {$(0,3)$} (3);
\draw [->] (1) to [bend left=40] node [sloped,etiquettemilieu] {$(0,0)$} (4);
\draw [->] (1) to [bend left=50] node [sloped,etiquettemilieu] {$(0,1)$} (5);
\draw [->] (1) to [bend left=17] node [sloped,etiquettemilieu] {$(1,2)$} (0);
\draw [->] (2) to [bend left=17] node [sloped,etiquettemilieu] {$(1,1)$} (3);
\draw [->] (2) to [bend left=30] node [sloped,etiquettemilieu] {$(1,2)$} (4);
\draw [->] (2) to [bend left=40] node [sloped,etiquettemilieu] {$(1,3)$} (5);
\draw [->] (3) to [bend left=40] node [sloped,etiquettemilieu] {$(2,0)$} (0);
\draw [->] (3) to [bend left=30] node [sloped,etiquettemilieu] {$(2,1)$} (1);
\draw [->] (3) to [bend left=17] node [sloped,etiquettemilieu] {$(2,2)$} (2);
\draw [->] (4) to [bend left=17] node [sloped,etiquettemilieu] {$(2,1)$} (5);
\draw [->] (4) to [bend left=50] node [sloped,etiquettemilieu] {$(3,2)$} (0);
\draw [->] (4) to [bend left=40] node [sloped,etiquettemilieu] {$(3,3)$} (1);
\draw [->] (5) to [bend left=45] node [sloped,etiquettemilieu] {$(3,0)$} (2);
\draw [->] (5) to [bend left=35] node [sloped,etiquettemilieu] {$(3,1)$} (3);
\draw [->] (5) to [bend left=17] node [sloped,etiquettemilieu] {$(3,2)$} (4);
\end{tikzpicture}
\vspace*{-0.5cm}
\caption{The automaton $\A_{6,4}$ accepts the language $\val_4^{-1}\big(\{(n,6n)\colon n\in \N\}\big)$.}
\label{fig:A-6,4}
\end{figure}
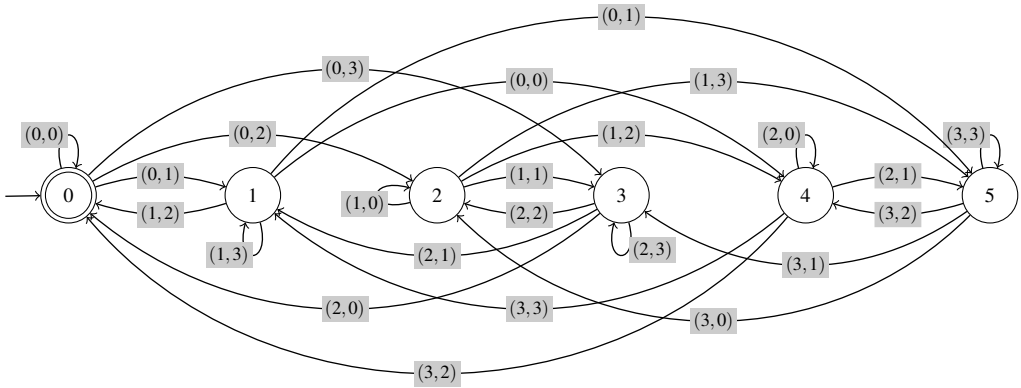

Note that the automaton $\A_{m,b}$ is not complete (see Remark~\ref{rem:unicitelettrepremcomp})
and has a loop labeled by $(0,0)$ on the initial state $0$.

\begin{remark}
\label{rem:unicitelettrepremcomp}
For each $i \in [\![0,m{-}1]\!]$ and $e \in A_b$, there exist unique $d \in A_b$ and $j \in [\![0,m{-}1]\!]$ such that $\delta_{m,b}(i,(d,e))=j$. 
\end{remark}

\begin{lemma} \label{lem:transitionsAmb}
For $i,j\in [\![0,m{-}1]\!]$ and $(u,v) \in (A_b\times A_b)^*$, we have
\[
	\delta_{m,b}(i,(u,v))=j \iff b^{|(u,v)|}\, i + \val_b(v) = m\, \val_b(u) + j.
\]
\end{lemma}

\begin{proof}
The proof is done by induction on $n=|(u,v)|$. 
\end{proof}

\begin{remark} \label{rem:unicitepremcomp}
It is easily checked that Remark~\ref{rem:unicitelettrepremcomp} extends from letters to words: for each $i \in [\![0,m{-}1]\!]$ and $v\in A_b^*$, there exist unique $u\in A_b^*$ and $j \in [\![0,m{-}1]\!]$ such that $\delta_{m,b}(i,(u,v))=j$. In particular, the word $u$ must have the same length as the word $v$, and hence $\val_b(u)<b^{|v|}$. 
\end{remark}

\begin{proposition} \label{prop:Ambproperties}
The automaton $\A_{m,b}$ is accessible, co-accessible and has disjoint states. 
\end{proposition}

\begin{proof}
Let $i \in [\![0,m{-}1]\!]$. From Lemma~\ref{lem:transitionsAmb}, we have $\delta_{m,b}(0, \rep_b(0,i))=i$. Therefore $\A_{m,b}$ is accessible. In order to find a word $(u,v)$ of some length $n$ that leads from $i$ to $0$, we consider the equation $b^n\, i + e = m d$ together with the constraints that $0\le d,e<b^n$. We can show that for any fixed $n$, such $d,e$ exist if and only if $\left\lceil \frac{b^ni}{m} \right\rceil - \frac{b^n}{m} < \frac{b^ni}{m} \le b^n-1$. Now take any $n$ satisfying these inequalities (it is always possible by choosing $n$ large enough). Then the word $0^{n-|\rep_b(d,e)|}\rep_b(d,e)$ is accepted from $i$, showing that $\A_{m,b}$ is co-accessible. Finally, let $j \in [\![0,m{-}1]\!]$. By Lemma~\ref{lem:transitionsAmb}, if $(u,v) \in L_i\cap L_j$ then $b^{|(u,v)|} i + \val_b(v) = m\, \val_b(u)$ and $b^{|(u,v)|} j + \val_b(v) = m\,\val_b(u)$, which implies that $i=j$. Thus $i\ne j\implies L_i\cap L_j=\emptyset$, i.e.\ $\A_{m,b}$ has disjoint states.
\end{proof}

In a reduced DFA, there can be at most one non co-accessible state. Thus, we deduce from Proposition~\ref{prop:Ambproperties} that $\A_{m,b}$ is indeed the {\em trim minimal} automaton of the language $\val_b^{-1}\big(\{(n,mn)\colon n\in \N\}\big)$, that is the automaton obtained by removing the only non co-accessible state from its minimal automaton.

\section{The projected automaton $\Pi(\A_{m,b})$}
\label{sec:projection-m}
In this section, we study the automaton $\Pi(\A_{m,b})$ obtained by projecting the label of each transition of $\A_{m,b}$ onto its second component. For each $i,j\in[\![0,m{-}1]\!]$, there is a transition labeled by $e\in A_b$ from the state $i$ to the state $j$ if and only if $j=bi+e\bmod m$.   

As is well known, the automaton $\Pi(\A_{m,b})$ is not minimal: it is minimal if and only if $m$ and $b$ are coprime; see for example \cite{Alexeev:2004}. In fact, whenever $m$ and $b$ are coprime, we have a stronger property than minimality as shown in the following proposition. This result will be useful in our future considerations.

\begin{proposition}\label{prop:projAmbdisj}
The automaton $\Pi(\A_{m,b})$ is complete, accessible and co-accessible. Moreover, if $m$ and $b$ are coprime, then the automaton $\Pi(\A_{m,b})$ has disjoint states, and hence it is the minimal automaton of $\val^{-1}_b(m\N)$.
\end{proposition}

\begin{proof}
The accessibility and co-accessibility of $\Pi(\A_{m,b})$ are straightforward consequences of Proposition~\ref{prop:Ambproperties}. Let $i,j\in[\![0,m{-}1]\!]$ and let $v \in A_b^*$ be a word accepted from both $i$ and $j$ in $\Pi(\A_{m,b})$. By Remark~\ref{rem:unicitepremcomp}, there exist unique words $u$ and $u'$ of the same length as $v$ such that $(u,v)$ and $(u',v)$ are accepted from $i$ and $j$ in $\A_{m,b}$ respectively. By Lemma~\ref{lem:transitionsAmb}, it is equivalent to say that $b^{|v|} i + \val_b (v) = m\, \val_b(u)$ and $b^{|v|} j + \val_b (v) = m\, \val_b (u')$.  Thus, we have
\begin{equation}\label{eqn:projAmbPropforte}
	b^{|v|} i - m\, \val_b (u) = b^{|v|} j -m\, \val_b( u').
\end{equation}
Therefore $m\, \val_b (u) \equiv m\, \val_b(u') \, \pmod{b^{|v|}}$. Because $m$ and $b$ are coprime, we obtain that $\val_b (u) \equiv \val_b(u') \, \pmod{b^{|v|}}$. Since $\val_b(u)$ and $\val_b (u')$ are both less than $b^{|v|}$, we obtain the equality $\val_b (u) =\val_b (u')$. Finally, we get from \eqref{eqn:projAmbPropforte} that $i=j$, which proves that $\Pi(\A_{m,b})$ has disjoint states.
\end{proof}

Let us prove some useful properties of $\Pi(\A_{m,b})$ under the more restrictive hypotheses of this work: $b=2^p$ and $m=k2^z$ with $k$ odd. 

\begin{lemma}
\label{lem:k>1}
If $k>1$ and $n=|\rep_{2^p}\big((k{-}1)2^z\big)|$, then $pn\ge z$.
\end{lemma}

For $k>1$ and $n=|\rep_{2^p}\big((k{-}1)2^z\big)|$, we let $\sigma$ be the permutation of the integers in $[\![0,k{-}1]\!]$ defined by $\sigma(j)=-j\,2^{pn-z}\bmod k$. Note that $\sigma$ permutes the integers $0,1,\ldots,k-1$ because $k$ is odd. For each $j\in [\![0,k{-}1]\!]$, we define $w_j$ to be the unique word of length $n$ representing $\sigma(j)2^z$ in base $2^p$:
\[
	w_j=0^{n-|\rep_{2^p}(\sigma(j)2^z)|}\rep_{2^p}(\sigma(j)2^z).
\]  
Note that the words $w_j$ are well defined since, by the choice of $n$, we have $\sigma(j)2^z\le (k{-}1)2^z<2^{pn}$ for every $j\in[\![0,k{-}1]\!]$.

\begin{proposition}
\label{prop:jj'}
Suppose that $k>1$ and let $j,j'\in[\![0,k{-}1]\!]$. Then the word $w_j$ is accepted from the state $j'$ in the automaton $\Pi(\A_{m,2^p})$ if and only if $j=j'$. 
\end{proposition}

\begin{proof}
Let $n=|\rep_{2^p}\big((k{-}1)2^z\big)|$. Then $|w_j|=n$ for all $j\in[\![0,k{-}1]\!]$ and from Lemma~\ref{lem:k>1}, we know that $pn\ge z$. We have
\begin{align*}
j' 2^{p|w_j|} +\val_{2^p}(w_j)\equiv 0\pmod m
		&\iff	j' 2^{pn} + \sigma(j)2^z\equiv 0\pmod{k2^z}\\
		&\iff 	j' 2^{pn-z} + \sigma(j)\equiv 0\pmod k \\
		&\iff 	j' 2^{pn-z} -j2^{pn-z} \equiv 0 \pmod k\\
		&\iff	j\equiv j' \pmod k \\
		&\iff	j=j'.
\end{align*}
\end{proof}

In a similar manner, we can prove the following result. 

\begin{proposition}
\label{prop:jj'-bis}
Suppose that $k>1$ and let $j,j'\in[\![0,k{-}1]\!]$. Then the word $w_j\rep_{2^p}(m)$ is accepted from the state $j'$ in the automaton $\Pi(\A_{m,2^p})$ if and only if $j=j'$. 
\end{proposition}

\section{The product automaton $\A_{m,2^p}\times \A_{\T,2^p}$}
\label{sec:product}

In this section, we study the product automaton $\A_{m,2^p} \times \A_{\T,2^p}$. The states of the product automaton are $(0,T), \ldots, (m{-}1,T)$ and  $(0,B) , \ldots , (m{-}1, B)$. The transitions of $\A_{m,2^p} \times \A_{\T,2^p}$ are defined as follows. For $i,j \in [\![0,m{-}1]\!]$, $X,Y \in \{T,B\}$ and $d,e \in A_{2^p}$, there is a transition labeled by $(d,e)$ from the state $(i,X)$ to the state $(j,Y)$ if and only if
\[
	2^p i + e = md + j \quad \andrm \quad 
	Y = 	X_d.
\]
We denote by $\delta_\times$ the (partial) transition function of this product automaton. The state $(0,T)$ is both initial and final, and there is no other final state.

From what precedes, namely Lemmas~\ref{lem:transitionsTM} and~\ref{lem:transitionsAmb}, Proposition~\ref{prop:Ambproperties} and the fact that $\A_{\T,2^p}$ has disjoint states, we obtain the following results.

\begin{lemma}
\label{lem:transitionsProd}
For all $i,j\in[\![0,m{-}1]\!]$, $X ,Y \in \{T,B\}$ and $(u,v)\in (A_{2^p}\times A_{2^p})^*$, we have
$\delta_\times((i,X),(u,v))=(j,Y)$ if and only if
\[
	2^{p\,|(u,v)|}\, i + \val_{2^p}(v) = m\, \val_{2^p}(u) + j 
	\quad \andrm \quad
	Y= X_{\val_{2^p}(u)}.
\]
\end{lemma}

\begin{corollary}
\label{cor:0Baccessible}
The word $\rep_{2^p}(1,m)$ is accepted from the state $(0,B)$ in $\A_{m,2^p} \times \A_{\T,2^p}$. In particular, the state $(0,B)$ is co-accessible in $\A_{m,2^p} \times \A_{\T,2^p}$.
\end{corollary}

\begin{lemma} \label{lem:proddisj} \ 
\begin{itemize}
\item For each $i \in [\![0,m{-}1]\!]$, the states $(i,T)$ et $(i,B)$ of the automaton $\A_{m,2^p} \times \A_{\T,2^p}$ are disjoint.
\item For distinct $i,j \in [\![0,m{-}1]\!]$ and for $X,Y \in\{T,B\}$, the states $(i,X)$ et $(j,Y)$ are disjoint in $\A_{m,2^p} \times \A_{\T,2^p}$.
\end{itemize}
\end{lemma}

We are now ready to establish the main properties of the product automaton $\A_{m,2^p} \times \A_{\T,2^p}$.

\begin{proposition}
\label{prop:prodaccessible}
The automaton $\A_{m,2^p} \times \A_{\T,2^p}$ is complete, accessible, co-accessible and has disjoint states. In particular, it is the minimal automaton of the language $\val^{-1}_{2^p}\left(\{(t, mt)\colon t \in \T\}\right)$.
\end{proposition}

\begin{proof}
By using Lemma~\ref{lem:transitionsProd}, we can verify that for every $i\in[\![0,m{-}1]\!]$, the states $(i,T)$ and $(i,B)$ are accessible thanks to the words $\rep_{2^p}(0,i)$ and $\rep_{2^p}(1,m+i)$ respectively. Hence, $\A_{m,2^p} \times \A_{\T,2^p}$ is accessible. To show the co-accessibility, we now fix some $i\in[\![0,m{-}1]\!]$ and $X \in \{T,B \}$. By Proposition~\ref{prop:Ambproperties}, we already know that the automaton $\A_{m,2^p}$ is co-accessible. Therefore, we can find $(u,v)\in(A_{2^p}\times A_{2^p})^*$ such that there is a path labeled by $(u,v)$ from $i$ to $0$ in $\A_{m,2^p}$. Thus, by reading $(u,v)$ from the state $(i,X)$ in $\A_{m,2^p} \times \A_{\T,2^p}$, we reach either the state $(0,T)$ or the state $(0,B)$. If we reach $(0,T)$, then the state $(i,X)$ is co-accessible. If we reach $(0,B)$ instead, then we may apply Corollary~\ref{cor:0Baccessible} in order to obtain that $(i,X)$ is co-accessible as well. Finally, we deduce from Lemma~\ref{lem:proddisj} that $\A_{m,2^p} \times \A_{\T,2^p}$ has disjoint states.
\end{proof}

\section{The projection $\Pi \left( \A_{m,2^p} \times \A_{\T,2^p} \right)$ of the product automaton}
\label{sec:projected-product}

The aim of this section is to provide a DFA accepting the language $\val^{-1}_{2^p} \left( m\T \right)$. This automaton is denoted by $\Pi \left( \A_{m,2^p} \times \A_{\T,2^p} \right)$ and is defined from 
$\A_{m,2^p} \times \A_{\T,2^p}$ by only keeping the second component in the label of each transition. Formally, the states of $\Pi \left( \A_{m,2^p} \times \A_{\T,2^p} \right)$ are $(0,T), \ldots , (m{-}1, T)$ and $(0,B) , \ldots , (m{-}1, B)$,
the state $(0,T)$ is both initial and final and no other state is final, and the transitions are defined as follows. For $i,j \in [\![0,m{-}1]\!]$, $X,Y \in \{T,B\}$ and $e \in A_{2^p}$, there is a transition labeled by $e$ from the state $(i,X)$ to the state $(j,Y)$ if and only if there exists $d \in A_{2^p}$ such that
\[
	2^p i + e = md + j \quad \andrm \quad 
	Y = 	X_d.
\]

\begin{example}
The automaton $\Pi \left( \A_{6,4} \times \A_{\T,4} \right)$ is depicted in Figure~\ref{fig:projected-aut}. All edges labeled by $0$ ($1,2$ and $3$ respectively) are represented in black (blue, red and green respectively). The colors of the states will become clear in Section~\ref{sec:minimization} and, in particular, in Example~\ref{ex:classes}.
\begin{figure}[htb]
\centering
\begin{tikzpicture}[scale=0.5]
\tikzstyle{every node}=[shape=circle, fill=none, draw=black,
minimum size=30pt, inner sep=2pt,scale=0.7]
\node(0T) at (0,0) {$0T$};
\node[fill=cyan](1T) at (3.5,0) {$1T$};
\node[fill=gray](2T) at (7,0) {$2T$};
\node[fill=yellow](3T) at (10.5,0) {$3T$};
\node[fill=magenta](4T) at (14,0) {$4T$};
\node[fill=orange](5T) at (17.5,0) {$5T$};
\node[fill=yellow](0B) at (0,-4.5) {$0B$};
\node[fill=magenta](1B) at (3.5,-4.5) {$1B$};
\node[fill=orange](2B) at (7,-4.5) {$2B$};
\node[fill=violet](3B) at (10.5,-4.5) {$3B$};
\node[fill=cyan](4B) at (14,-4.5) {$4B$};
\node[fill=gray](5B) at (17.5,-4.5) {$5B$};
\tikzstyle{every node}=[shape=circle, fill=none, draw=black,
minimum size=25pt, inner sep=2pt,scale=0.7]
\node at (0,0) {};

\tikzstyle{etiquettedebut}=[very near start,rectangle,fill=black!20,scale=0.7]
\tikzstyle{etiquettemilieu}=[midway,rectangle,fill=black!20,scale=0.7]
\tikzstyle{every path}=[color=black, line width=0.5 pt]
\tikzstyle{every node}=[shape=circle, minimum size=5pt, inner sep=2pt]
\draw [->] (-1.5,0) to node {} (0T); 
\draw [->] (0T) to [loop above] node [scale=0.7] {$0$} (0T);
\draw [green,->] (5T) to [loop above] node [] {} (5T);
\draw [->] (0B) to [loop below] node [] {} (0B);
\draw [green,->] (5B) to [loop below] node [] {} (5B);
\draw [blue,->] (0T) to [] node [above=-0.05,scale=0.7] {$1$} (1T);
\draw [red,->] (0T) to [bend left=20] node [above=-0.057,scale=0.7] {$2$} (2T);
\draw [green,->] (0T) to [bend left=25] node [above,scale=0.7] {$3$} (3T);
\draw [->] (1T) to [bend left=20] node [] {} (4T);
\draw [blue,->] (1T) to [bend left=30] node [] {} (5T);
\draw [red,->] (1T) to [] node [] {} (0B);
\draw [green,->] (1T) to [bend left=15] node [] {} (1B);
\draw [->] (2T) to [bend left=15] node [] {} (2B);
\draw [blue,->] (2T) to [bend left=5] node [] {} (3B);
\draw [red,->] (2T) to [] node [] {} (4B);
\draw [green,->] (2T) to [] node [] {} (5B);
\draw [->] (3T) to [] node [] {} (0B);
\draw [blue,->] (3T) to [] node [] {} (1B);
\draw [red,->] (3T) to [bend left=5] node [] {} (2B);
\draw [green,->] (3T) to [bend left=15] node [] {} (3B);
\draw [->] (4T) to [bend left=15] node [] {} (4B);
\draw [blue,->] (4T) to [] node [] {} (5B);
\draw [red,->] (4T) to [bend right=35] node [] {} (0T);
\draw [green,->] (4T) to [bend right=25] node [] {} (1T);
\draw [->] (5T) to [bend right=25] node [] {} (2T);
\draw [blue,->] (5T) to [bend right=20] node [] {} (3T);
\draw [red,->] (5T) to [] node [] {} (4T);
\draw [blue,->] (0B) to [] node [] {} (1B);
\draw [red,->] (0B) to [bend right=20] node [] {} (2B);
\draw [green,->] (0B) to [bend right=25] node [] {} (3B);
\draw [->] (1B) to [bend right=20] node [] {} (4B);
\draw [blue,->] (1B) to [bend right=30] node [] {} (5B);
\draw [red,->] (1B) to [] node [] {} (0T);
\draw [green,->] (1B) to [bend left=15] node [] {} (1T);
\draw [->] (2B) to [bend left=15] node [] {} (2T);
\draw [blue,->] (2B) to [bend left=5] node [] {} (3T);
\draw [red,->] (2B) to [] node [] {} (4T);
\draw [green,->] (2B) to [] node [] {} (5T);
\draw [->] (3B) to [] node [] {} (0T);
\draw [blue,->] (3B) to [] node [] {} (1T);
\draw [red,->] (3B) to [bend left=5] node [] {} (2T);
\draw [green,->] (3B) to [bend left=15] node [] {} (3T);
\draw [->] (4B) to [bend left=15] node [] {} (4T);
\draw [blue,->] (4B) to [] node [] {} (5T);
\draw [red,->] (4B) to [bend left=35] node [] {} (0B);
\draw [green,->] (4B) to [bend left=25] node [] {} (1B);
\draw [->] (5B) to [bend left=25] node [] {} (2B);
\draw [blue,->] (5B) to [bend left=20] node [] {} (3B);
\draw [red,->] (5B) to [] node [] {} (4B);
\end{tikzpicture}
\caption{The projected automaton $\Pi \left( \A_{6,4} \times \A_{\T,4} \right)$.}
\label{fig:projected-aut}
\end{figure}
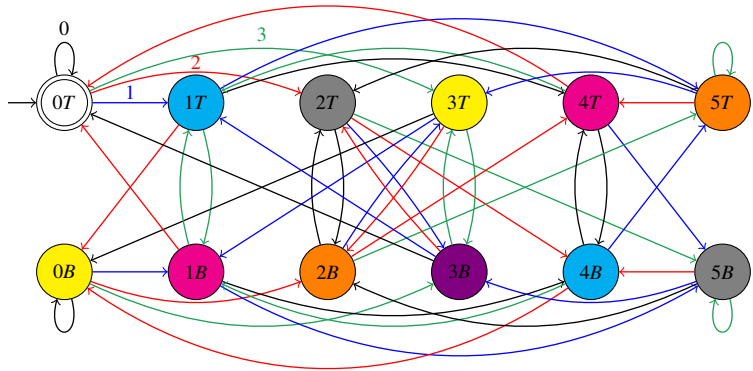
\end{example}

\begin{lemma}
\label{lem:projdisj}
For each $i \in [\![0,m{-}1]\!]$, the states $(i,T)$ and $(i,B)$ are disjoint. 
\end{lemma}

\begin{proof}
This follows from Remark~\ref{rem:unicitelettrepremcomp} and Lemma~\ref{lem:proddisj}. 
\end{proof}

\begin{proposition}
\label{prop:m-odd}
The automaton $\Pi \left( \A_{m,2^p} \times \A_{\T,2^p} \right)$ accepts $\val^{-1}_{2^p} \left( m \T \right)$, is deterministic, complete, accessible and co-accessible. Moreover, if $m$ is odd then it has disjoint states, and hence is minimal.
\end{proposition}

\begin{proof}
By construction, $\Pi \left( \A_{m,2^p} \times \A_{\T,2^p} \right)$ accepts $\val^{-1}_{2^p} \left( m \T \right)$; see Section~\ref{sec:method}. The fact that this automaton is deterministic and complete follows from Remark~\ref{rem:unicitelettrepremcomp}. It is accessible and co-accessible because so is $\A_{m,2^p} \times \A_{\T,2^p}$.  If a word $v$ over $A_{2^p}$ is accepted from some state $(i,X)$ in $\Pi \left( \A_{m,2^p} \times \A_{\T,2^p} \right)$,  then there exists a word $u$ over $A_{2^p}$ of length $|v|$ such that the word $(u,v)$ is accepted from $(i,X)$ in $\A_{m,2^p} \times \A_{\T,2^p}$. We deduce that $(u,v)$ is accepted from the state $i$ in $\A_{m,2^p}$ and in turn, that $v$ is accepted from the state $i$ in $\Pi \left( \A_{m,2^p} \right)$. Therefore, and by combining Proposition~\ref{prop:projAmbdisj} and Lemma~\ref{lem:projdisj}, we obtain  that if $m$ is odd then the automaton $\Pi \left( \A_{m,2^p} \times \A_{\T,2^p} \right)$ has disjoint states. It directly follows that $\Pi \left( \A_{m,2^p} \times \A_{\T,2^p} \right)$ is minimal if $m$ is odd.
\end{proof}

\begin{corollary}
\label{cor:m-odd}
If $m$ is odd, then the state complexity of $m \T $ with respect to the base $2^p$ is $2m$.
\end{corollary}

Note that Corollary~\ref{cor:m-odd} and Theorem~\ref{thm:main} are consistent in the case where $m$ is odd, i.e.\  where $z=0$. However, we will see in Theorem~\ref{thm:main2} that the DFA $\Pi \left( \A_{m,2^p} \times \A_{\T,2^p} \right)$ is never minimal for even $m$. 


\section{Minimization of $\Pi \left( \A_{m,2^p} \times \A_{\T,2^p} \right)$}
\label{sec:minimization}

We start by defining some classes of states of $\Pi \left( \A_{m,2^p} \times \A_{\T,2^p} \right)$. Our aim is twofold. First, we will prove that those subsets consist in {\em indistinguishable} states, i.e.\ accepting the same language. Second, we will show that states belonging to different such subsets are {\em distinguishable}, i.e.\ accepts different languages. Otherwise stated, these classes correspond to the left quotients $w^{{-}1}L$ where $w$ is any word over the alphabet $A_{2^p}$ and $L=\val^{-1}_{2^p}(m\T)$.

\begin{definition}
For $(j,X) \in \big([\![1,k{-}1]\!]\times \{T,B\}\big)\cup \{(0,B)\}$, we define the {\em classes}
\[
		[(j,X)] = \{ ( j+k\ell, X_\ell) \colon \ell\in [\![0,2^z{-}1]\!]\} \quad\andrm\quad
[(0,T)] = \{ (0,T)\}.
\] 
\end{definition}

Note that these classes are pairwise disjoint: $[(j,X)]\cap[(j',X')]=\emptyset$ if $(j,X)\ne (j',X')$. If $m$ is odd, i.e.\ if $z =0$, then all these classes are reduced to a single state. If $m$ is a power of $2$, i.e.\  if $k=1$, then there is no class of the form $[(j,X)]$ with $j\ge 1$. 

\begin{definition}
For $\alpha \in [\![0,z{-}1]\!]$, we define the {\em pre-classes} 
\[
	C_\alpha = \{( k2^{z-\alpha - 1} + k 2^{z - \alpha} \ell , B_\ell) \colon  \ell\in [\![0,2^\alpha{-}1]\!]\}.
\]
Then, for $\beta \in [\![0,\lceil\frac zp \rceil{-}2]\!]$,  we define the {\em classes}
\[
	\Gamma_\beta = \bigcup_{\alpha = \beta p}^{\beta p+ p - 1} C_\alpha.
\]
In addition, we set
\[
	\Gamma_{\lceil \frac zp \rceil - 1} 
	= \bigcup_{\alpha = \left(\lceil \frac zp \rceil - 1 \right) p}^{z{-}1} C_{\alpha}.
\]
\end{definition}

\begin{remark}
Note that the classes $\Gamma_\beta$ are pairwise disjoint. If $m$ is odd, i.e.\ if $z =0$, then there is no such class $\Gamma_\beta$. 
\end{remark}

\begin{remark}
If a class $[(j,X)]$ or $\Gamma_\beta$ exists, then it is nonempty. Moreover, the classes $\Gamma_\beta$ together with the class $[(0,T)]$ form a partition of $ \{(k\ell,T_\ell) \colon \ell\in[\![0,2^z{-}1]\!]\}$. Therefore, the classes $[(j,X)]$ and $\Gamma_\beta$ form a partition of the set of states of $\Pi \left( \A_{m,2^p} \times \A_{\T,2^p} \right)$.
\end{remark}


\begin{example}
\label{ex:classes}
For $m=6$ and $p=2$, it is easily verified that the classes defined above correspond to states of the same color in the automaton $\Pi \left( \A_{6,4} \times \A_{\T,4} \right)$ of Figure~\ref{fig:projected-aut}. Let us make the explicit computations for $m=24$ and $p=2$: in this case, the classes defined above are
\begin{align*}
	[(0,T)] & =\{(0,T)\} \\
	[(1,T)] & =\{(1,T),(4,B),(7,B),(10,T),(13,B),(16,T),(19,T),(22,B)\} \\
	[(2,T)] & =\{(2,T),(5,B),(8,B),(11,T),(14,B),(17,T),(20,T),(23,B)\} \\
	[(0,B)] & =\{(0,B),(3,T),(6,T),(9,B),(12,T),(15,B),(18,B),(21,T)\} \\
	[(1,B)] & =\{(1,B),(4,T),(7,T),(10,B),(13,T),(16,B),(19,B),(22,T)\} \\
	[(2,B)] & =\{(2,B),(5,T),(8,T),(11,B),(14,T),(17,B),(20,B),(23,T)\} \\
	\Gamma_0 & =C_0\cup C_1= \{(12,B)\}\cup\{(6,B),(18,T)\}=\{(6,B),(12,B),(18,T)\} \\
	\Gamma_1 & =C_2= \{(3,B),(9,T),(15,T),(21,B)\}.
\end{align*}
In Figure~\ref{fig:projected-aut24}, the states of the automaton $\Pi \left( \A_{24,4} \times \A_{\T,4} \right)$ are colored with respect to these classes.
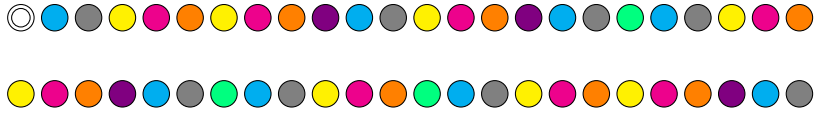
\begin{figure}
\centering
\begin{tikzpicture}
[>=triangle 60,scale=0.225]
\tikzstyle{every node}=[shape=circle, fill=none, draw=black,
minimum size=10pt, inner sep=2pt]
\node(0T) at (0,0) {};
\node[fill=cyan](1T) at (2,0) {};
\node[fill=gray](2T) at (4,0) {};
\node[fill=yellow](3T) at (6,0) {};
\node[fill=magenta](4T) at (8,0) {};
\node[fill=orange](5T) at (10,0) {};
\node[fill=yellow](6T) at (12,0) {};
\node[fill=magenta](7T) at (14,0) {};
\node[fill=orange](8T) at (16,0) {};
\node[fill=violet](9T) at (18,0) {};
\node[fill=cyan](10T) at (20,0) {};
\node[fill=gray](11T) at (22,0) {};
\node[fill=yellow](12T) at (24,0) {};
\node[fill=magenta](13T) at (26,0) {};
\node[fill=orange](14T) at (28,0) {};
\node[fill=violet](15T) at (30,0) {};
\node[fill=cyan](16T) at (32,0) {};
\node[fill=gray](17T) at (34,0) {};
\node[fill=vert](18T) at (36,0) {};
\node[fill=cyan](19T) at (38,0) {};
\node[fill=gray](20T) at (40,0) {};
\node[fill=yellow](21T) at (42,0) {};
\node[fill=magenta](22T) at (44,0) {};
\node[fill=orange](23T) at (46,0) {};

\node[fill=yellow](0B) at (0,-4.5) {};
\node[fill=magenta](1B) at (2,-4.5) {};
\node[fill=orange](2B) at (4,-4.5) {};
\node[fill=violet](3B) at (6,-4.5) {};
\node[fill=cyan](4B) at (8,-4.5) {};
\node[fill=gray](5B) at (10,-4.5) {};
\node[fill=vert](6B) at (12,-4.5) {};
\node[fill=cyan](7B) at (14,-4.5) {};
\node[fill=gray](8B) at (16,-4.5) {};
\node[fill=yellow](9B) at (18,-4.5) {};
\node[fill=magenta](10B) at (20,-4.5) {};
\node[fill=orange](11B) at (22,-4.5) {};
\node[fill=vert](12B) at (24,-4.5) {};
\node[fill=cyan](13B) at (26,-4.5) {};
\node[fill=gray](14B) at (28,-4.5) {};
\node[fill=yellow](15B) at (30,-4.5) {};
\node[fill=magenta](16B) at (32,-4.5) {};
\node[fill=orange](17B) at (34,-4.5) {};
\node[fill=yellow](18B) at (36,-4.5) {};
\node[fill=magenta](19B) at (38,-4.5) {};
\node[fill=orange](20B) at (40,-4.5) {};
\node[fill=violet](21B) at (42,-4.5) {};
\node[fill=cyan](22B) at (44,-4.5) {};
\node[fill=gray](23B) at (46,-4.5) {};

\tikzstyle{every node}=[shape=circle, fill=none, draw=black,
minimum size=7pt, inner sep=2pt]
\node at (0,0) {};
\end{tikzpicture}
\caption{The classes of the projected automaton $\Pi \left( \A_{24,4} \times \A_{\T,4} \right)$.}
\label{fig:projected-aut24}
\end{figure}
\end{example}

\subsection{States of the same class are indistinguishable}
\label{sec:reduction1}

For any two states $(j,X)$ and $(j',X')$ of the projected automaton $\Pi \left( \A_{m,2^p} \times \A_{\T,2^p} \right)$, the general procedure that we use for proving that $L_{(j,X)}\subseteq L_{(j',X')}$  goes as follows. Let $v\in L_{(j,X)}$ and let $n=|v|$. There exists a word $u$ over $A_{2^p}$ of length $|v|$ such that $(u,v)$ is accepted from the state $(j,X)$ in $\A_{m,2^p} \times \A_{\T,2^p}$ (before the projection). If $d=\val_{2^p}(u)$ and $e = \val_{2^p} (v)$, then, in view of Lemma~\ref{lem:transitionsProd}, we must have 
\[
	2^{pn}j+e = md \quad \andrm \quad X_d=T
\]
(the only final state of $\A_{m,2^p} \times \A_{\T,2^p}$ is $(0,T)$). Moreover, since $n=|v|$, we have $d,e\in[\![0,2^{pn}{-}1]\!]$. Now, in order to prove that $v\in L_{(j',X')}$, we have to find a word $u'$ over $A_{2^p}$ of length $n$ such that $(u',v)$ is accepted from $(j',X')$ in $\A_{m,2^p} \times \A_{\T,2^p}$. But then, we necessarily have that 
\[
	\val_{2^p}(u')=\frac{2^{pn}j'+e}{m}.
\] 
Let thus $d'=\frac{2^{pn}j'+e}{m}$. We obtain that $v\in L_{(j',X')}$ if and only if $d'\in[\![0,2^{pn}{-}1]\!]$ and $X'_{d'}=T$. Indeed, in this case, $|\rep_{2^p}(d')|\le n$ and thus, we can take the word $u'=0^{n-|\rep_{2^p}(d')|}\rep_{2^p}(d')$.
\medskip

We show that two states of the same class are indistinguishable. We give part of the proof of the first proposition only. 

\begin{proposition}
\label{prop:jX}
Let $(j,X) \in \big([\![1,k{-}1]\!]\times \{T,B\}\big)\cup\{(0,B)\}$ and let $\ell \in [\![1, 2^z{-}1]\!]$. Then $L_{( j,X)} = L_{( j+k\ell, X_\ell)}$.
\end{proposition}

\begin{proof}
We only give the proof for $j\ge 1$. The proof for $(0,B)$ can be adapted from this one.  Let $v\in A_{2^p}^*$, $n=|v|$, $e=\val_{2^p}(v)$, $d=\frac{2^{pn}j+e}{m}$ and $d'=\frac{2^{pn}(j+k\ell)+e}{m}$. We have to prove that $d\in[\![0,2^{pn}{-}1]\!]$ and $X_d=T$ if and only if $d'\in[\![0,2^{pn}{-}1]\!]$ and $(X_\ell)_{d'}=T$. 

Since $j\in[\![1,k{-}1[\!]$ and $e\in[\![0,2^{pn}{-}1]\!]$, we have
\begin{equation}
\label{eqn:jX}
	0< d=\frac{2^{pn}j+e}{m}<\frac{2^{pn}k}{m}=2^{pn-z}.
\end{equation}
Since $d'= d + \frac{2^{pn}k\ell}{m} =d+2^{pn-z}\ell$, it follows from~\eqref{eqn:jX} that if $d$ and $d'$ are both integers, then we must have
\[
	\rep_2(d')=\rep_2(\ell)0^{pn-z-|\rep_2(d)|}\rep_2(d).
\] 
Therefore, $d\in\T$ if and only if either $\ell\in\T\andrm d'\in\T$, or $\ell\notin\T\andrm d'\notin\T$, and hence  $X_d=(X_\ell)_{d'}$. 

Now, suppose that $d\in[\![0,2^{pn}{-}1]\!]$ and $X_d=T$. It follows from~\eqref{eqn:jX} that $pn> z$, for otherwise we would have $0<d<1$, which is not possible since $d$ is an integer. Therefore, we get that $d' = d +2^{pn-z}\ell$ is a positive integer. We also get from~\eqref{eqn:jX} that
\[
	d' = d+2^{pn-z}\ell<2^{pn-z} (\ell+1)\le  2^{pn}.
\]
Consequently, $d'\in[\![0,2^{pn}{-}1]\!]$ and $(X_\ell)_{d'}=X_d=T$. 

Conversely, suppose that $d'\in[\![0,2^{pn}{-}1]\!]$ and $(X_\ell)_{d'}=T$. In view of~\eqref{eqn:jX} and since $d= d'-2^{pn-z}\ell$, in order to obtain that $d\in[\![0,2^{pn}{-}1]\!]$, it is enough to show that $pn> z$. Proceed by contradiction and suppose that $pn \le z$. Let $q=\big\lfloor\frac{\ell}{2^{z-pn}}\big\rfloor$. On the one hand, since $j\ge 1$ and $e\ge 0$, we obtain
\[
	d'=\frac{2^{pn}(j+k\ell)+e}{m}>\frac{2^{pn}k\ell}{m}=\frac{\ell}{2^{z-pn}}\ge q.
\]
On the other hand, since $\ell\le (q+1)2^{z-pn}{-}1$, $e<2^{pn}$ and $j\le k{-}1$, we obtain
\[
	d' 	< \frac{2^{pn}(j+k(q+1)2^{z-pn}-k)+2^{pn}}{m} 
		= q+1+2^{pn}\frac{j-k+1}{m} 
		\le q+1.
\]
This is not possible since $d'$ is an integer, and hence $pn> z$. Consequently, $d\in[\![0,2^{pn}{-}1]\!]$ and $X_d=T$ as desired. 
\end{proof}

\begin{corollary}\label{cor:0B}
For each $(j,X)\in \big([\![1,k{-}1]\!]\times\{T,B\}\big)\cup\{(0,B)\}$, all states of the class $[(j,X)]$ are indistinguishable.
\end{corollary}

Similarly, we can prove that two states of the same class of the form $\Gamma_\beta$ are indistinguishable.

\begin{proposition}\label{prop:gammabeta}
For all $\beta \in[\![0,\big\lceil \frac zp \big\rceil {-}1]\!]$, all states of the class $\Gamma_\beta$ are indistinguishable. 
\end{proposition}

\subsection{States of different classes are distinguishable}
\label{sec:reduction2}

In this section, we show that, in the projected automaton $\Pi \left( \A_{m,2^p} \times \A_{\T,2^p} \right)$, states from different classes $[(j,X)]$ or $\Gamma_\beta$ are pairwise distinguishable, that is, for any two such states, there exists a word which is accepted from exactly one of them.

First of all, note that the state $(0,T)$ is distinguished from all other states since it is the only final state: the empty word $\varepsilon$ is accepted from $(0,T)$ but not from any other state.

\begin{proposition}
Let $\beta,\gamma \in[\![0,\big\lceil \frac zp \big\rceil {-}1]\!]$ such that $\gamma > \beta$. The word $0^{\beta +1}$ is accepted from all states of $\Gamma_{\beta}$ but is not accepted from any state of $\Gamma_{\gamma}$.
\end{proposition}

\begin{proof}
From Proposition~\ref{prop:gammabeta}, it suffices to show that $0^{\beta +1}$ is accepted from the state $( k2^{z - \gamma p {-}1},B)$ if and only if $\gamma=\beta$. This is an easy verification.
\end{proof}

\begin{proposition}
Let $(j,X) \in\big([\![1,k{-}1]\!]\times \{T,B\}\big)\cup\{(0,B)\}$ and $\beta \in[\![0,\big\lceil \frac zp \big\rceil {-}1]\!]$. The word $0^{\beta +1}$ is not accepted from any state of $[(j,X)]$.
\end{proposition}

\begin{proof}
Since there is a loop labeled by $0$ on the state $(0,T)$ and in view of Proposition~\ref{prop:jX}, it suffices to show that the word $0^{\lceil z/p\rceil}$ is not accepted from the state $(j,X)$. If $0^{\lceil z/p\rceil}$ were accepted from the state $(j,X)$, then we would get that
\[
	d=\frac{2^{p\left\lceil\frac zp\right\rceil}j}{m}=\frac{2^{p\left\lceil\frac zp\right\rceil-z}j}{k}
\]
is an integer such that $X_d=T$. If $j\ne 0$ then $d$ cannot be an integer since $k$ is odd and $0<j<k$. If $j=0$ then $X=X_0=X_d=T$, which contradicts the assumption that $(j,X)\ne (0,T)$. Hence the conclusion.
\end{proof}

\begin{proposition}
Suppose that $k>1$ and let $(j,X),(j',X')\in\big([\![1,k{-}1]\!]\times \{T,B\}\big)\cup\{(0,B)\}$ be distinct. 
The states $(j,X)$ and $(j',X')$ are distinguishable.
\end{proposition}

\begin{proof}
Suppose that $j=j'$. Then $X\ne X'$ by hypothesis and the states $(j,X)$ and $(j,X')$ are disjoint by Lemma~\ref{lem:projdisj}. Since $\Pi \left(\A_{m,2^p} \times \A_{\T,2^p}  \right)$ is co-accessible by Proposition~\ref{prop:m-odd}, we obtain that the states $(j,X)$ and $(j,X')$ are distinguishable.

Now suppose that $j\ne j'$. By Proposition~\ref{prop:jj'}, the word $w_j$ is accepted from $j$ in the automaton $\Pi(\A_{m,2^p})$ but is not accepted from $j'$. Then, there exists a word $u$ of length $|w_j|$ such that $(u,w_j)$ is accepted from $j$ in the automaton $\A_{m,2^p}$ but is not accepted from $j'$. Then, this word $(u,w_j)$ is accepted either from $(j,T)$ or from $(j,B)$ in the automaton $\A_{m,2^p}\times \A_{\T,2^p}$ but is not accepted neither from $(j',T)$ nor from $(j',B)$. Now, two cases are possible. 

First, suppose that $(u,w_j)$ is accepted from $(j,X)$ in $\A_{m,2^p}\times \A_{\T,2^p}$. Then, in the projected automaton $\Pi \left(\A_{m,2^p} \times \A_{\T,2^p}  \right)$, the word $w_j$ is accepted from $(j,X)$ but not from $(j',X')$. Thus, the word $w_j$ distinguishes the states $(j,X)$ and $(j',X')$. 

Second, suppose that $(u,w_j)$ is accepted from $(j,\overline{X})$ in $\A_{m,2^p}\times \A_{\T,2^p}$. Then there is a path labeled by $(u,w_j)$ from $(j,X)$ to $(0,B)$ in $\A_{m,2^p}\times \A_{\T,2^p}$. By Corollary~\ref{cor:0Baccessible}, in $\A_{m,2^p}\times \A_{\T,2^p}$, the word $\rep_{2^p}(1,m)$ is accepted from $(0,B)$, and hence the word $(u,w_j)\rep_{2^p}(1,m)=(u0^{|\rep_{2^p}(m)|{-}1}1,w_j\rep_{2^p}(m))$ is accepted from $(j,X)$. Therefore the word $w_j\rep_{2^p}(m)$ is accepted from the state $(j,X)$ in $\Pi \left(\A_{m,2^p} \times \A_{\T,2^p}  \right)$. Besides, the word $w_j\rep_{2^p}(m)$ cannot be accepted from $(j',X')$ in $\Pi \left(\A_{m,2^p} \times \A_{\T,2^p}  \right)$ for otherwise it would also be accepted from $j'$ in $\Pi \left(\A_{m,2^p} \right)$, which is impossible by Proposition~\ref{prop:jj'-bis}. Thus, the word $w_j\rep_{2^p}(m)$ distinguishes the states $(j,X)$ and $(j',X')$. 
\end{proof}

\subsection{The minimal automaton of $\val^{-1}_{2^p} (m \T)$.}
We are ready to construct the minimal automaton of $\val^{-1}_{2^p}(m\T)$. Since the states of $\Pi \left( \A_{m,2^p} \times \A_{\T,2^p} \right)$ that belong to the same class $[(j,X)]$ or $\Gamma_\beta$ are indistinguishable, they can be glued together in order to define a new automaton $\mathcal{M}_{m,\T,2^p}$ that still accepts the same language. Formally, the alphabet of  $\mathcal{M}_{m,\T,2^p}$ is $A_{2^p}$. Its states are the classes $[(j,X)]$ for $(j,X)\in[\![0,k{-}1]\!]\times \{T,B\}$ and the classes $\Gamma_\beta$ for $\beta\in[\![0,\lceil \frac zp\rceil{-}1]\!]$. The class $[(0,T)]$ is the initial state and the only final state. The transitions of $\mathcal{M}_{m,\T,2^p}$ are defined as follows: there is a transition labeled by a letter $a$ in $A_{2^p}$ from a class $J_1$ to a class $J_2$ if and only if there exists $j_1 \in J_1$ and $j_2 \in J_2$ such that, in the automaton $\Pi \left( \A_{m,2^p}\times \A_{\T,2^p} \right)$, there is a transition labeled by $a$ from the state $j_1$ to the state $j_2$. 

From what precedes, we obtain the following result.

\begin{theorem}
\label{thm:main2}
Let $p$ and $m$ be positive integers. The automaton $\mathcal{M}_{m,\T,2^p}$ is the minimal automaton of the language $\val^{-1}_{2^p} (m\T)$. 
\end{theorem}

Note that Proposition~\ref{prop:m-odd} and Theorem~\ref{thm:main2} are consistent in the case where $m$ is odd, i.e.\ where $z=0$.

We are now ready to prove Theorem~\ref{thm:main}.

\begin{proof}[Proof of Theorem~\ref{thm:main}]
In view of Theorem~\ref{thm:main2}, it suffices to count the number of states of $\mathcal{M}_{m,\T,2^p}$. By definition, it has $2(k{-}1)+2=2k$ states of the form $[(j,X)]$ and $\lceil \frac{z}{p}\rceil$ states of the form $\Gamma_\beta$.
\end{proof}

As an application of this result, we obtain the following decision procedure.

\begin{corollary}
Given any $2^p$-recognizable set $Y$ (via a finite automaton $\A$ recognizing it), it is decidable whether $Y=m\T$ for some $m\in\N$. The decision procedure can be run in time $O(N^2)$ where $N$ is the number of states of the given automaton $\A$.
\end{corollary}

\begin{proof}
Let $Y$ be a $2^p$-recognizable set given thanks to a finite (complete) automaton that accepts the languages of the $2^p$-expansions of its elements. Let $N$ be the number of states of this automaton. Then we can minimize and hence compute the state complexity $M$ of $Y$ (with respect to the base $2^p$) in time $O(N\log(N))$ \cite{Hopcroft:1971}. Let us decompose the possible multiples $m$ as $k2^z$ with $k$ odd. By Theorem~\ref{thm:main}, it is sufficient to test the equality between $Y$ and $m\T$ for the finitely many values of pairs $(k,z)$ such that $2k+\lceil\frac{z}{p}\rceil=M$. Since $M\le N$, the number of such tests is in $O(N)$. For each $m$ that has to be tested, we can directly use our description of the minimal automaton of $\val_{2^p}^{-1}(m\T)$ (this is Theorem~\ref{thm:main2}). This concludes the proof since the equality of two regular languages is decidable in linear time \cite{Hopcroft&Karp:2015}.
\end{proof}

%

\section{Conclusion and perspectives}
\label{sec:perspectives}

Our method is constructive and general: in principle, it may be applied to any $b$-recognizable set $X\subseteq\N$. However, in general, it is not the case that the product automaton $\A_{m,2^p} \times \A_{X,2^p}$ recognizing the bidimensional set $\{(n,mn)\colon n\in X\}$ is minimal. As an example, consider the $2$-recognizable set $X$ of powers of $2$: $X=\{2^n\colon n\in\N\}$. Then the product automaton $\A_{3,2} \times \A_{X,2}$ of our construction (for $m=3$ and $b=2$) has $6$ states but is clearly not minimal since it is easily checked that the automaton of Figure~\ref{fig:powers-2} is the minimal automaton recognizing the set $\{(2^n,3\cdot 2^n)\colon n\in  \N\}$.
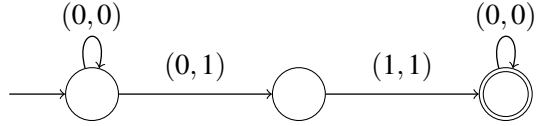
\begin{figure}
\centering
\begin{tikzpicture}[scale=0.55]
\tikzstyle{every node}=[shape=circle, fill=none, draw=black,minimum size=20pt]
\node(1) at (0,0) {};
\node(2) at (5,0) {};
\node(3) at (10,0) {};
\tikzstyle{every node}=[shape=circle, fill=none, draw=black,minimum size=17pt]
\node at (10,0) {};

\tikzstyle{every path}=[color=black, line width=0.5 pt]
\tikzstyle{every node}=[shape=circle, minimum size=5pt, inner sep=2pt]

\draw [->] (-2,0) to node {} (1); 
\draw [->] (1) to [loop above] node [above=-0.3cm] {$(0,0)$} (1);
\draw [->] (3) to [loop above] node [above=-0.3cm] {$(0,0)$} (3);
\draw [->] (1) to [] node [above=-0.2cm] {$(0,1)$} (2);
\draw [->] (2) to [] node [above=-0.2cm] {$(1,1)$} (3);
\end{tikzpicture}
\caption{Minimal automaton recognizing the set $\{(2^n,3\cdot 2^n)\colon n\in\N\}$.}
\label{fig:powers-2}
\end{figure}
This illustrates that, in general, the minimization procedure is not only needed in the final projection $\Pi \left( \A_{m,2^p} \times \A_{X,2^p} \right)$ as is the case in the present work.

Nevertheless, we conjecture that the phenomenon described in this work for the Thue-Morse set also appears for all $b$-recognizable sets of the form 
\[	
	X_{b,c,M,R}=\{n\in\N\colon |\rep_b(n)|_c\equiv R\bmod M\}
\]
where $b$ is an integer base, $c$ is any digit in $A_b$, $M$ is an integer greater than or equal to $2$ and $R$ is any possible remainder in $[\![0,M-1]\!]$. More precisely, we conjecture that whenever the base $b$ is a prime power, i.e.\ $b=q^p$ for some prime $q$, then the state complexity of $mX_{b,c,M,R}$ is given by the formula $Mk+\lceil \frac{z}{p}\rceil$ where $k$ is the part of the multiple $m$ that is prime to the base $b$, i.e.\ $m=kq^z$ with $\gcd(k,q)=1$.
 
We end by mentioning two other potential future research directions in the continuation of the present work. The first is to consider automata reading the expansions of numbers with least significant digit first. Both reading directions are relevant to different problems. For example, it is easier to compute addition thanks to an automaton reading expansions from “right to left” than from “left to right”. On the opposite, if we have in mind to generalize our problems to $b$-recognizable sets of real numbers (see for instance \cite{Boigelot&Rassart&Wolper:1998,Charlier:2018,Charlier&Leroy&Rigo:2015}), then the relevant reading direction is the one with most significant digit first. 
Further, there is no intrinsic reason why the state complexity from “left to right” should be the same as (or even close to) that obtained from “right to left”. 
The second related problem we want to investigate is the computation of the state complexity of the operation $X\mapsto m X+r$ where $r$ is not necessarily equal to $0$ as is the case in this work. We conjecture that the state complexity will be the same for all $r\in[\![0,m-1]\!]$.

\section{Acknowledgment}
Célia Cisternino is supported by the FNRS Research Fellow grant 1.A.564.19F.

\bibliographystyle{eptcs}
\bibliography{TMmultiples}
\label{sec:biblio}
\end{document}